\documentclass[aps,prl,reprint,amsmath,amssymb,longbibliography]{revtex4-2}

\usepackage{lipsum}
\usepackage{amsmath}
\usepackage{amssymb}
\usepackage{amsthm}
\usepackage{mathrsfs}
\usepackage{accents}
\usepackage{calc}
\usepackage{upgreek}
\usepackage{slashed}
\usepackage{xifthen}
\usepackage{graphicx}
\usepackage{longtable}
\usepackage[inline]{enumitem}

\usepackage{tikz}

\usepackage{xcolor}
\definecolor{winered}{rgb}{0.6,0,0}
\definecolor{lessblue}{rgb}{0,0,0.7}

\usepackage[pdftex,colorlinks=true,linkcolor=winered,citecolor=lessblue,urlcolor=lessblue,breaklinks=true,bookmarksopen=true]{hyperref}

\makeatletter
\newcommand{\myitem}[2]{\item[\rm(#2)]\def\@currentlabel{#2}\label{#1}}
\makeatother

\makeatletter
\def\@nomenstarted{0}
\newlength{\@nomenoldtabcolsep}

\newcommand{\nomenstart}
  {%
    \def\@nomenstarted{1}%
    \setlength{\@nomenoldtabcolsep}{\tabcolsep}%
    \setlength{\tabcolsep}{3.5pt}%
    \begin{longtable}{p{0.11\textwidth} p{0.86\textwidth}}
  }

\newcommand{\nomenitem}[2]{%
    \ifcase\@nomenstarted%
      \or 
      \or \\ 
    \fi%
    #1\,{\leavevmode\leaders\hbox{\,.}\hfill} & #2%
    \def\@nomenstarted{2}%
  }%
\newcommand{\nomenend}
  {\\%
      \end{longtable}%
      \setlength{\tabcolsep}{\@nomenoldtabcolsep}%
      \def\@nomenstarted{0}%
  }
\makeatother

\makeatletter
\newcommand{\BIG}{\bBigg@{3.5}}
\newcommand{\vast}{\bBigg@{4}}
\newcommand{\Vast}{\bBigg@{5}}
\newcommand{\VAST}[1]{\bBigg@{#1}}
\makeatother

\allowdisplaybreaks

\newtheoremstyle{prl}{\topsep}{\topsep}{\normalfont}{\parindent}
  {\itshape}{---}{0.3em}{}
\theoremstyle{prl}

\newtheorem{thm}{Theorem}

\newtheorem{lemma}[thm]{Lemma}

\newtheorem*{thm*}{Theorem}
\newtheorem*{prop*}{Proposition}
\newtheorem*{cor*}{Corollary}
\newtheorem*{conj*}{Conjecture}

\theoremstyle{definition}

\theoremstyle{remark}

\makeatletter
\newcommand{\fakephantomsection}{%
  \Hy@MakeCurrentHref{\@currenvir.\the\Hy@linkcounter}
  \Hy@raisedlink{\hyper@anchorstart{\@currentHref}\hyper@anchorend}%
  \Hy@GlobalStepCount\Hy@linkcounter%
}
\makeatother

\newcommand{\mc}{\mathcal}
\newcommand{\cA}{\mc A}

\newcommand{\cC}{\mc C}
\newcommand{\cD}{\mc D}

\newcommand{\cH}{\mc H}
\newcommand{\cI}{\mc I}
\newcommand{\cJ}{\mc J}

\newcommand{\cL}{\mc L}

\newcommand{\cO}{\mc O}

\newcommand{\cS}{\mc S}

\newcommand{\cV}{\mc V}

\newcommand{\ms}{\mathscr}

\newcommand{\sD}{\ms D}

\newcommand{\sS}{\ms S}

\newcommand{\sW}{\ms W}

\newcommand{\C}{\mathbb{C}}
\newcommand{\N}{\mathbb{N}}
\newcommand{\R}{\mathbb{R}}
\newcommand{\Z}{\mathbb{Z}}

\newcommand{\Sph}{\mathbb{S}}

\newcommand{\fs}{\mathfrak{s}}

\newcommand{\bOne}{{\boldsymbol{1}}}
\renewcommand{\bell}{{\boldsymbol{\ell}}}
\renewcommand{\bm}{{\boldsymbol{m}}}
\newcommand{\bn}{{\boldsymbol{n}}}

\renewcommand{\Re}{\operatorname{Re}}
\renewcommand{\Im}{\operatorname{Im}}

\newcommand{\mathspan}{\operatorname{span}}
\newcommand{\supp}{\operatorname{supp}}

\newcommand{\tr}{\operatorname{tr}}

\DeclareMathOperator{\adj}{adj}
\newcommand{\vecz}{{\!\vec{\;z}}}

\newcommand{\hra}{\hookrightarrow}
\newcommand{\la}{\langle}

\renewcommand{\ol}{\overline}
\newcommand{\pa}{\partial}
\newcommand{\dd}{{\mathrm d}}
\newcommand{\ra}{\rangle}

\newcommand{\wt}{\widetilde}

\newcommand{\bop}{{\mathrm{b}}}

\renewcommand{\cp}{{\mathrm{c}}}

\newcommand{\CI}{\cC^\infty}

\newcommand{\CIc}{\cC^\infty_\cp}

\newcommand{\openbigpmatrix}[1]
  {%
    \def\@bigpmatrixsize{#1}%
    \addtolength{\arraycolsep}{-#1}%
    \begin{pmatrix}%
  }
\newcommand{\closebigpmatrix}
  {%
    \end{pmatrix}%
    \addtolength{\arraycolsep}{\@bigpmatrixsize}%
  }

\newlength{\enummargin}

\DeclareGraphicsExtensions{.mps}

\makeatletter
\newcommand*{\fwbw}[1]{\expandafter\@fwbw\csname c@#1\endcsname}
\newcommand*{\@fwbw}[1]{\ifcase #1 \or {\rm fw}\or {\rm bw}\fi}
\AddEnumerateCounter{\fwbw}{\@fwbw}
\makeatother

\begin{document}
\raggedbottom

\title{Perturbations of Kerr--Newman black holes: separation and mode stability}

\date{\today}


\author{Peter Hintz}
\email[Contact author: ]{phintz@psu.edu}
\affiliation{Department of Mathematics, Pennsylvania State University, 54 McAllister St, State College, PA 16801, United States}

\begin{abstract}
  We study linearized perturbations of the subextremal Kerr--Newman black hole solution of the Einstein--Maxwell equations. We show how to overcome what Chandrasekhar called the ``apparent indissolubility of the coupling between the spin-1 and spin-2 fields'' by using a matrix integrating factor for suitably transformed first-order systems that describe the radiative degrees of freedom. For the radial ODEs obtained after angular separation, we prove mode stability on the real axis, i.e., the absence of quasinormal modes with real frequencies, using a Whiting-type transform.
\end{abstract}

\maketitle

\newcommand{\prlheading}[1]{\par\textit{#1}---\ignorespaces}
\prlheading{Introduction}%
The Kerr--Newman (KN) solution \cite{NewmanCouchChinnaparedExtonPrakashTorrenceKN} of the Einstein--Maxwell equations describes a stationary asymptotically flat black hole (BH) in electrovacuum; it is the unique such solution \cite{BuntingKNUniqueness,MazurKerrNewmanUniqueness,ChruscielCostaHeuslerStationaryBH,LaiLiYuKNRigidity}. It depends on the parameters $M>0$ (mass), $a M$ (angular momentum), $Q$ (electric charge). The study of its perturbations is a classical topic \cite{ChitreKNPert,ChandrasekharBlackHoles}: in the Newman--Pen\-rose (NP) formalism \cite{NewmanPenroseSpin}, radiative degrees of freedom of linearized perturbations ${\sim}\,e^{-i\omega t+i m\phi}$ that are purely outgoing at future null infinity and purely ingoing at the future event horizon (hereafter called ``outgoing'') are described by $2\times 4$ scalars solving a \emph{positive spin system} (coupled first-order PDEs~\eqref{EqEPlus1}--\eqref{EqEPlus4} in $(r,\theta)$) and a \emph{negative} spin system.

In the Kerr \cite{KerrKerr} limit, the perturbation equations separate into the spin $\fs\in\{\pm 2,\pm 1\}$ Teukolsky equations \cite{TeukolskyKerr,TeukolskySeparation}, which can be fully decoupled into angular and radial ODEs. Mode stability of subextremal Kerr for the Einstein\mbox{(--Maxwell)} equations, i.e., the absence of outgoing mode solutions for $\Im\omega>0$, was proved by Whiting \cite{WhitingKerrModeStability}. Recent developments include \cite{ShlapentokhRothmanModeStability,AnderssonMaPaganiniWhitingModeStab,TeixeiradCModes,AminovGrassiHatsudaQNM,CasalsTeixeiradCModes,AnderssonHaefnerWhitingMode,PetersenVasyModeStab,ShlapentokhRothmanTeixeiradCTeukolskyI,ShlapentokhRothmanTeixeiradCTeukolskyII} and the proof of nonlinear stability of subextremal Kerr \cite{HintzKerrStab} (extending earlier work \cite{KlainermanSzeftelPolarized,AnderssonBackdahlBlueMaKerr,HaefnerHintzVasyKerr,KlainermanSzeftelKerr,GiorgiKlainermanSzeftelStability,ShenGCMKerr,DafermosHolzegelRodnianskiTaylorSchwarzschild,HintzPetersenVasyKdS,HaefnerHintzVasyKerrLarge}), with mode stability as a key input.

Either spin system is essentially equivalent to a coupled system of second-order PDEs, whose separation on KN with $a,Q\neq 0$ had remained elusive \cite{ChandrasekharBlackHoles,GiorgiKNEquations,KalninsWilliamsKN}. We resolve this issue by working with the first-order system: after transformation into a symmetric ``chiral'' form~\eqref{EqPf1Chiral}, we identify a matrix integrating factor~\eqref{EqPf2X} that produces a system~\eqref{EqPf2Final} amenable to radial and angular separation. Separated solutions of~\eqref{EqPf2Final} are described by an angular eigenvalue problem~\eqref{EqPf3AngEig} and a radial ODE~\eqref{EqPf3ODE}. A variant of the Whiting transform (inspired by \cite{PetersenVasyModeStab}) transforms the radial problem to an ODE amenable to standard Wronskian arguments for mode stability on the real axis (i.e., for real $\omega$). (In \cite[\S{3.9}]{HintzKdSMS}, \cite{HintzKdSFullModes}, mode stability on the real axis is shown to imply full ($\Im\omega\geq 0$) mode stability for scalar waves on Kerr: the continuous dependence of putative unstable modes on the BH parameters, a priori bounds on $|\Re\omega|,\Im\omega$ from WKB or energy methods, and the well-known full mode stability of Schwarzschild imply that unstable modes for Kerr would have to arise from real axis crossings---which mode stability on the real axis prevents.)

Detailed numerical studies \cite{DiasGodazgarSantosKerrNewman,DiasGodazgarSantosKerrNewmanExt} gave evidence for full mode stability up to $99.999\%$ of extremality for $|m|\leq l=1,2,3,4$ modes, in addition to information on quasinormal modes (QNMs). Evidence for nonlinear stability is given in \cite{ZilhaoCardosoHerdeiroLehnerSperhakeKNStab}. Semi-analytic works include \cite{PaniBertiGualtieriSlowKN,PaniBertiGualtieriSlowKN2,MarkYangZimmermanChenQNMSlowKN}. The linear stability of KN for $|\frac{Q}{M}|,|\frac{a}{M}|\ll 1$ was proved in \cite{HeLinearKerrNewman}; see also \cite{GiorgiKNTeukolsky}, further \cite{GiorgiRNLinearSmall,GiorgiRNLinear,MoncriefRNOdd,MoncriefRNEven,MoncriefRNGaugeInv} for the case $a=0$, and \cite{GiorgiWanKNAxial,MoncriefGudapatiKNAxi} regarding KN in axisymmetry. While QNMs are not discussed here, our decoupled system may be helpful for their further study: unlike, e.g., the Dudley--Finley approximation \cite{DudleyFinleyApprox,KokkotasKNQNMs,BertiKokkotasKNQNMs,BertiCardosoStarinetsQNM,ZimmermanMarkZeroDampedQNM,SahaSilvaKNDudleyFinely}, it is \emph{exact} for all subextremal KN parameters with $Q\neq 0$ and for all complex frequencies $\omega\in\C$.

\prlheading{Definitions and equations}%
We consider a subextremal KN black hole, so $a^2+Q^2<M^2$. Set
\[
  \Delta=r^2{-}2 M r{+}a^2{+}Q^2,\,
  \bar\rho=r{+}i a\cos\theta,\,
  \bar\rho^*=r{-}i a\cos\theta.
\]
Thus, $\Delta$ has roots $0\leq r_-<r_+=M{+}\sqrt{M^2{-}a^2{-}Q^2}$. We work in the Kinnersley tetrad
\begin{align*}
  \bell^\mu\pa_\mu &= \tfrac{1}{\Delta}\bigl((r^2+a^2)\pa_t+\Delta\pa_r+a\pa_\phi\bigr), \\
  \bn^\mu\pa_\mu &= \tfrac{1}{2\bar\rho\bar\rho^*}\bigl((r^2+a^2)\pa_t-\Delta\pa_r+a\pa_\phi\bigr), \\
  \bm^\mu\pa_\mu &= \tfrac{\sin\theta}{\bar\rho\sqrt 2}\bigl(i a\pa_t+\tfrac{1}{\sin\theta}\pa_\theta+\tfrac{i}{\sin^2\theta}\pa_\phi\bigr).
\end{align*}
The metric signature convention $(+,-,-,-)$ corresponds to $\bell\cdot\bn=1$, $\bm\cdot\bar\bm=-1$. We write $e_1=\bell$, $e_2=\bn$, $e_3=\bm$, $e_4=\bar\bm$. Our convention for the curvature tensor is $R(e_\mu,e_\nu)e_\rho=([\nabla_{e_\mu},\nabla_{e_\nu}]-\nabla_{[e_\mu,e_\nu]})e_\rho=R^\sigma{}_{\rho\mu\nu}e_\sigma$; the Weyl tensor $C_{\sigma\rho\mu\nu}$ is its trace-free part. The middle Weyl component on KN is $\Psi_2=-C_{1 3 4 2}=\frac{Q^2-M\bar\rho}{\bar\rho\bar\rho^{*3}}$. The NP spin coefficients are $\gamma_{i j k}=(e_i)^\nu{}_{;\mu}(e_j)_\nu(e_k)^\mu$.

The relevant Weyl scalars and NP spin coefficients for us are $\Psi_0=-C_{1 3 1 3}$, $\Psi_1=-C_{1 2 1 3}$, $\Psi_3=-C_{1 2 4 2}$, $\Psi_4=-C_{2 4 2 4}$ and $\kappa=-\gamma_{3 1 1}$, $\sigma=-\gamma_{3 1 3}$, $\lambda=\gamma_{4 2 4}$, $\nu=\gamma_{4 2 2}$; they vanish on KN. For the Maxwell scalars, we use the convention $\Phi_0=F_{1 3}$, $\Phi_1=\frac12(F_{1 2}+F_{4 3})$, $\Phi_2=F_{4 2}$; on KN, $\Phi_0=\Phi_2=0$ and $\Phi_1=\frac{Q}{2\bar\rho^{*2}}$.

We consider perturbations of a subextremal KN solution $g,F=\dd A$. Linearized quantities carry the superscript ``$(1)$''. We treat $Q\neq 0$ (the Kerr case being known). We recall the phantom gauge \cite{ChandrasekharBlackHoles}: under null rotations around $\bell$, i.e., $\bell'=\bell$, $\bm'=\bm+\bar b\bell$, $\bn'=\bn+b\bm+\bar b\bar\bm+|b|^2\bell$, where $b$ is a complex-valued function that vanishes on the KN background, we have $(\Phi_2^{(1)})'=\Phi_2^{(1)}+2 b^{(1)}\Phi_1$. We can thus set this to $0$ for an appropriate choice of $b^{(1)}$ for $Q\neq 0$. Similarly, null rotations around $\bn$ allow us to set $(\Phi_0^{(1)})'=0$. (The quantity $\Phi_1^{(1)}$ is unaffected, and typically nonzero, but does not enter in the radiative subsystem studied in the paper.) We drop the prime and thus continue with $\Phi_0^{(1)}=\Phi_2^{(1)}=0$.

Radiative degrees of freedom are now described by the linearized Weyl scalars $\Psi_j^{(1)}$, $j=0,1,3,4$; following \cite{ChandrasekharBlackHoles,MarkYangZimmermanChenQNMSlowKN,DiasGodazgarSantosKerrNewmanExt}, we rescale them as \footnote{See \cite[\S{2.2}]{DiasGodazgarSantosKerrNewmanExt} for a gauge-invariant interpretation. Division by $Q$ appears also there \cite[(2.22)]{DiasGodazgarSantosKerrNewmanExt}.}
\[
  \psi_{-2}{=}\bar\rho^{*4}\Psi_4^{(1)},\,
  \psi_{-1}{=}\tfrac{\bar\rho^{*3}\Psi_3^{(1)}}{\sqrt 2},\,
  \psi_1{=}\sqrt{2}\bar\rho^*\Psi_1^{(1)},\,
  \psi_2{=}\Psi_0^{(1)}.
\]
We study mode solutions. By axisymmetry, it suffices to study perturbations of the form $e^{-i\omega t+i m\phi}u(r,\theta)$, defined for $r>r_+$, $\theta\in(0,\pi)$, with $\omega\in\R$, $m\in\Z$. Set
\begin{alignat*}{2}
  K &:= -(r^2+a^2)\omega+a m, &\ \hat Q &:= -a\omega\sin\theta + \tfrac{m}{\sin\theta}, \\
  \cD_j &:= \pa_r + \tfrac{i K}{\Delta} + j\tfrac{\Delta'}{\Delta}, &\ \cL_j&:=\pa_\theta+\hat Q+j\cot\theta, \\
  \cD_j^\dag &:= \pa_r - \tfrac{i K}{\Delta} + j\tfrac{\Delta'}{\Delta}, &\ \cL_j^\dag&:=\pa_\theta-\hat Q+j\cot\theta.
\end{alignat*}
In terms of the rescaled linearized NP spin coefficients
\[
  k=\tfrac{\kappa^{(1)}}{\sqrt{2}\bar\rho^{*2}},\quad
  s=\tfrac{\bar\rho\sigma^{(1)}}{\bar\rho^{*2}},\quad
  \ell=\tfrac{\bar\rho^*\lambda^{(1)}}{2},\quad
  n=\tfrac{\bar\rho\bar\rho^*\nu^{(1)}}{\sqrt{2}},
\]
the linearized NP equations can be expressed as two sets of four coupled first-order equations: setting $F_\pm=3(M-\frac{Q^2}{\bar\rho})\pm\frac{Q^2\bar\rho^*}{\bar\rho^2}$, we call (see \cite[Ch.~11, eq.~(139)--(142)]{ChandrasekharBlackHoles})
\begin{subequations}
\begin{align}
\label{EqEPlus1}
  \Bigl(\cL_2{-}\frac{3 i a\sin\theta}{\bar\rho^*}\Bigr)\psi_2 {-} \Bigl(\cD_0{+}\frac{3}{\bar\rho^*}\Bigr)\psi_1 &= {-}2 k F_+, \\
\label{EqEPlus2}
  \Delta\Bigl(\cD_2^\dag{-}\frac{3}{\bar\rho^*}\Bigr)\psi_2 {+} \Bigl(\cL_{{-}1}^\dag{+}\frac{3 i a\sin\theta}{\bar\rho^*}\Bigr)\psi_1 &= 2 s F_-, \\
\label{EqEPlus3}
  \Bigl(\cD_0{+}\frac{3}{\bar\rho^*}\Bigr)s {-} \Bigl(\cL_{{-}1}^\dag{+}\frac{3 i a\sin\theta}{\bar\rho^*}\Bigr)k &= \frac{\bar\rho}{\bar\rho^{*2}}\psi_2, \\
\label{EqEPlus4}
  \Delta\Bigl(\cD_2^\dag{-}\frac{3}{\bar\rho^*}\Bigr)k {+} \Bigl(\cL_2{-}\frac{3 i a\sin\theta}{\bar\rho^*}\Bigr)s &= \frac{2\bar\rho}{\bar\rho^{*2}}\psi_1,
\end{align}
\end{subequations}
the \emph{positive spin system}, and the second, obtained using the replacements $\psi_1\to-\psi_{-1}$, $\psi_2\to\psi_{-2}$, $k\to -n$, $s\to\ell$, $\cL_{-1}^\dag\to\cL_{-1}$, $\cL_2\to\cL_2^\dag$, $\cD_0+\frac{3}{\bar\rho^*}\to\Delta(\cD_{-1}^\dag+\frac{3}{\bar\rho^*})$, $\Delta(\cD_2^\dag-\frac{3}{\bar\rho^*})\to\cD_0-\frac{3}{\bar\rho^*}$, the \emph{negative spin system} \footnote{The positive (and likewise the negative) spin system can, to some extent, be reduced to a closed system of two coupled second-order scalar equations for $(\psi_2,\psi_1)$; a caveat is that this reduction involves division by $F_-$ in~\eqref{EqEPlus2}, which can vanish in the KN exterior for certain KN parameters (e.g., $a=0$ and $\frac{Q^2}{M^2}>\frac{15}{16}$). Giorgi \cite[Theorem~6.1]{GiorgiKNEquations} avoids this issue by not performing such a complete reduction and instead retaining an enlarged system (that, moreover, remains regular as $Q\to 0$). By contrast, we avoid this issue by working directly with regular first-order systems, albeit only for $Q\neq 0$. Conversely, the second-order system implies the first-order system upon \emph{defining} $k,s$ by~\eqref{EqEPlus1}--\eqref{EqEPlus2}.}.

\prlheading{Main result}%
{\bf Theorem.} On a subextremal KN black hole with $Q\neq 0$, all solutions of the positive spin system~\eqref{EqEPlus1}--\eqref{EqEPlus4} (and likewise of the negative spin system) with $\omega\in\R$, $m\in\Z$ that are \emph{outgoing} (and, in the case $\omega=0$, obeying the bounds below) must vanish.\hfill$\blacksquare$

The outgoing condition is as follows: setting $H(r):=\exp(i\int^r_{2 r_+} K/\Delta\,\dd r)$, the normalized quantities
\begin{equation}
\label{EqMainHor}
  H^{-1}\Delta^b\psi_b,\,
  H^{-1}\Delta^2 k,\,
  H^{-1}\Delta s,\,
  H^{-1}\Delta^{-2}n,\,
  H^{-1}\Delta^{-1}\ell
\end{equation}
are smooth on $[r_+,\infty)$. (The origin is that $H(r)$ captures the purely ingoing asymptotics at the event horizon, and $\psi_b$ has boost-weight $b$, i.e., scales with $\Delta^b$ upon scaling $\bell,\bn\rightsquigarrow\Delta\bell,\Delta^{-1}\bn$; similarly for $k,s,n,\ell$ with boost-weights $b=2,1,-2,-1$.)

At infinity and for $\omega\neq 0$, we demand that $H\psi_b,H k,H s,H n,H\ell=\cO(r^C)$ for some fixed $C$; by $f=f(r,\theta)=\cO(r^C)$ we mean that $|(r\pa_r)^j\pa_\theta^k f|\lesssim_{j,k}r^C$ for all $j,k$ \footnote{Using~\eqref{EqEPlus1}--\eqref{EqEPlus4} and the negative spin system, the stated outgoing condition for $\omega\neq 0$ turns out to force precise asymptotics for $\psi_b,k,s,n,\ell$ by a simple asymptotic analysis of the system~\eqref{EqEPlus1}--\eqref{EqEPlus4}: the quantities with boost-weight $b$ are of class $e^{i\omega(r+2 M\log r)}r^{-1-2 b}\times\CI(1/r,\theta)$. We do not need this more precise information for our arguments, however.}. For $\omega=0$, we require $\psi_\fs=\cO(r^{-1-\fs})$; the two spin systems then give $\psi_2,k=\cO(r^{-3})$, $\psi_1,s=\cO(r^{-2})$, $\psi_{-1},\ell=\cO(1)$, $\psi_{-2},n=\cO(r)$.

Near the pole $\theta=0$ (resp.\ $\theta=\pi$) of $\Sph^2$, we demand that for a quantity $f$ with spin-weight $\fs$, $e^{i\fs\phi}f$ (resp.\ $e^{-i\fs\phi}f$) extends smoothly. (This arises from the observation that $e^{i\phi}\bm$, resp.\ $e^{-i\phi}\bm$, extends smoothly to the pole.) The spin weights of $\psi_\fs,k,s,n,\ell$ are $\fs,1,2,-1,-2$. For smooth modes $f\sim e^{i m\phi}$, this implies $|f|\lesssim\theta^{|m+\fs|}$ (resp.\ $|f|\lesssim(\pi-\theta)^{|m-\fs|}$).

\prlheading{Proof of the Theorem}%
We refer to the positive (negative) spin system using $\varsigma=+1$ ($-1$). We treat $\omega\neq 0$ in Steps~1--4 below, and $\omega=0$ in Step~5.

\paragraph{Step~1: transformation into chiral form.}

The principal part of~$\Delta^{-\frac12}\cdot$\eqref{EqEPlus2}, $(-1)\cdot$\eqref{EqEPlus1}, resp.\ $\Delta^{-\frac12}\cdot$\eqref{EqEPlus4}, \eqref{EqEPlus3} is $\Bigl(\begin{smallmatrix} \Delta^{\frac12}\pa_r & \pa_\theta \\ -\pa_\theta & \Delta^{\frac12}\pa_r \end{smallmatrix}\Bigr)$ acting on $(\psi_2,\Delta^{-\frac12}\psi_1)$, resp.\ $(k,\Delta^{-\frac12}s)$. Diagonalizing gives the 2D Dirac-type operators $\Delta^{\frac12}\pa_r\mp i\pa_\theta$ acting on $\psi_2\pm i\Delta^{-\frac12}\psi_1$ and $k\pm i\Delta^{-\frac12}s$. Arguing similarly for the negative spin system, let us introduce
\begin{equation}
\label{EqPfypm}
\begin{split}
  y_\pm^{(+)} &:= \Delta^{\frac54} \sqrt{\sin\theta} \begin{pmatrix} \psi_2 \pm i\Delta^{-\frac12}\psi_1 \\ -2 Q(k\pm i\Delta^{-\frac12}s) \end{pmatrix}, \\
  y_\pm^{(-)} &:= \Delta^{-\frac34} \sqrt{\sin\theta} \begin{pmatrix} \psi_{-2} \mp i\Delta^{\frac12}\psi_{-1} \\ 2 Q(n \mp i\Delta^{\frac12}\ell) \end{pmatrix}.
\end{split}
\end{equation}
Then $y_\pm^{(\varsigma)}$ solves the system
\begin{gather}
\label{EqPf1Chiral}
  (\Delta^{\frac12}\pa_r\mp i\pa_\theta\mp A)y_\pm^{(\varsigma)} = (c^{(\varsigma)} I \mp B^{(\varsigma)})y_\mp^{(\varsigma)}, \\
  A = \begin{pmatrix} 0 & -i\frac{Q\bar\rho^*}{\bar\rho^2} \\ i\frac{Q\bar\rho}{\bar\rho^{*2}} & 0 \end{pmatrix}, \ 
  B^{(\varsigma)} = i\begin{pmatrix} b^{(\varsigma)} & \frac{3(M\bar\rho-Q^2)}{Q\bar\rho} \\ \frac{3 Q\bar\rho}{\bar\rho^{*2}} & -b^{(\varsigma)} \end{pmatrix}, \nonumber\\
  b^{(\varsigma)} {=} -\varsigma\hat Q {-} \tfrac32\cot\theta {+} \tfrac{3 i a\sin\theta}{\bar\rho^*},\,
  c^{(\varsigma)} {=} \Bigl(\varsigma\tfrac{i K}{\Delta}{+}\tfrac{3}{\bar\rho^*}{-}\tfrac{3\Delta'}{4\Delta}\Bigr)\Delta^{\frac12}. \nonumber
\end{gather}
The radial and angular factors $\Delta^{\frac14+\varsigma}$ and $\sqrt{\sin\theta}$ in~\eqref{EqPfypm} are chosen so that $\tr A=0$. The factor $2 Q$ ensures $A=A^*$.

\paragraph{Step~2: a matrix integrating factor.}

The $2\times 2$-matrix
\begin{equation}
\label{EqPf2X}
  X(r,\theta) =
    \begin{pmatrix}
      Q\frac{\Delta^{\frac12}+a\sin\theta}{\bar\rho} & -i\bigl(M-\frac{Q^2}{\bar\rho}\bigr) \\
      i\bigl(M-\frac{Q^2}{\bar\rho^*}\bigr) & Q\frac{\Delta^{\frac12}-a\sin\theta}{\bar\rho^*}
    \end{pmatrix},
\end{equation}
which is invertible since $\det X = Q^2-M^2 < 0$, satisfies
\begin{equation}
\label{EqPf2Int}
  (\Delta^{\frac12}\pa_r{-}i\pa_\theta)X{=}A X,\,
  (\Delta^{\frac12}\pa_r{+}i\pa_\theta)(X^*)^{-1}{=}{-}A(X^*)^{-1}.
\end{equation}
The second equation follows from the first and $A^*=A$.

The author discovered the matrix integrating factor $X$ in eq.~\eqref{EqPf2X} using the following observation. Write $\sD=\Delta^{\frac12}\pa_r-i\pa_\theta$. For the functions $f_1=\bar\rho$, $f_2=\Delta^{\frac12}-a\sin\theta$, $f_3=\bar\rho^*$, $f_4=\Delta^{\frac12}+a\sin\theta$, we then have
\[
  \sD f_1=f_2,\ \sD f_2=f_1-M,\quad \sD f_3=f_4,\ \sD f_4=f_3-M.
\]
Writing $(x,y)^T$ for either of the two columns of $X$, eq.~\eqref{EqPf2Int} reads $\sD x=-i Q\frac{f_3}{f_1^2}y$, $\sD y=i Q\frac{f_1}{f_3^2}x$. Applying $\sD$ to the first equation and using the second equation gives
\[
  \sD^2 x = \tfrac{Q^2}{f_1 f_3}x + \bigl(\tfrac{f_4}{f_3}-2\tfrac{f_2}{f_1}\bigr)\sD x.
\]
With the ansatz $x=f_1^{-1}(\sum_{j=1}^4\alpha_j f_j+\beta)$ for constants $\alpha_j,\beta\in\C$, one finds two linearly independent solutions $(\alpha_1,\alpha_2,\alpha_3,\alpha_4,\beta)=(0,0,0,1,0)$, $(M,0,0,0,-Q^2)$. These give the two columns of eq.~\eqref{EqPf2X}, up to convenient scalar factors.

Let now
\begin{equation}
\label{Eqzpm}
\begin{gathered}
  z_+^{(\varsigma)} := (\det X)X^{-1}y_+^{(\varsigma)},\quad
  z_-^{(\varsigma)} := X^*y_-^{(\varsigma)}, \\
  C^{(\varsigma)}(r,\theta) := (\det X)X^{-1}(c^{(\varsigma)}I-B^{(\varsigma)})(X^*)^{-1}.
\end{gathered}
\end{equation}
Then $(\Delta^{\frac12}\pa_r-i\pa_\theta)z_+^{(\varsigma)} = C^{(\varsigma)}z_-^{(\varsigma)}$, and from $\tr B^{(\varsigma)}=0$ we obtain $\adj(c^{(\varsigma)}I-B^{(\varsigma)})=c^{(\varsigma)}I+B^{(\varsigma)}$ (where $\adj\left(\begin{smallmatrix}a&b\\c&d\end{smallmatrix}\right)=\left(\begin{smallmatrix}d&-b\\-c&a\end{smallmatrix}\right)$ is the adjugate matrix), so $(\Delta^{\frac12}\pa_r+i\pa_\theta)z_-^{(\varsigma)} = \adj(C^{(\varsigma)})z_+^{(\varsigma)}$. Crucially, $C^{(\varsigma)}$ separates:
\begin{align*}
  &C^{(\varsigma)}(r,\theta) = R^{(\varsigma)}(r) + i\Theta^{(\varsigma)}(\theta); \\
  &\quad R^{(\varsigma)} = \begin{pmatrix} R_0^{(\varsigma)} & R_1^{(\varsigma)} \\ -R_1^{(\varsigma)} & R_0^{(\varsigma)} \end{pmatrix},\ 
    \Theta^{(\varsigma)} = \begin{pmatrix} \Theta_0^{(\varsigma)} & \Theta_1^{(\varsigma)} \\ \Theta_1^{(\varsigma)} & -\Theta_0^{(\varsigma)} \end{pmatrix}, \\
  &\quad R_0^{(\varsigma)} = \tfrac{3\Delta'}{4\Delta^{\frac12}} - i\varsigma\Delta^{\frac12}\bigl(\tfrac{K}{\Delta} - \tfrac{2\omega Q^2}{M^2-Q^2}\bigr), \\
  &\quad R_1^{(\varsigma)} = -\varsigma\tfrac{2\omega Q(M r-Q^2)}{M^2-Q^2} + \tfrac{3 i M}{2 Q}, \\
  &\quad \Theta_0^{(\varsigma)} = \tfrac32\cot\theta + \varsigma\bigl(\tfrac{m}{\sin\theta}-\tfrac{a\omega(M^2+Q^2)\sin\theta}{M^2-Q^2}\bigr), \\
  &\quad \Theta_1^{(\varsigma)} = -\varsigma\tfrac{2 a\omega M Q\cos\theta}{M^2-Q^2}-\tfrac{3 M}{2 Q}.
\end{align*}
Since $\adj C^{(\varsigma)}=(R^{(\varsigma)})^T-i\Theta^{(\varsigma)}$, $\vec z^{(\varsigma)}=\bigl(\begin{smallmatrix} z_+^{(\varsigma)}\\z_-^{(\varsigma)}\end{smallmatrix}\bigr)$ solves
\begin{equation}
\begin{split}
\label{EqPf2Final}
  &\bigl(\Delta^{\frac12}\pa_r - (R_0^{(\varsigma)} \cI + R_1^{(\varsigma)} \cJ) - \cA^{(\varsigma)}\bigr)\vecz^{(\varsigma)} = 0, \\
  &\quad
  \boxed{\cA^{(\varsigma)} := \begin{pmatrix} i\pa_\theta & i\Theta^{(\varsigma)} \\ -i\Theta^{(\varsigma)} & -i\pa_\theta \end{pmatrix},}
\end{split}
\end{equation}
where $\cI = \bigl(\begin{smallmatrix} 0 & I_{2\times 2} \\ I_{2\times 2} & 0 \end{smallmatrix}\bigr)$, $\cJ = \bigl(\begin{smallmatrix} 0 & -J \\ J & 0 \end{smallmatrix}\bigr)$, $J=\bigl(\begin{smallmatrix} 0 & -1 \\ 1 & 0 \end{smallmatrix}\bigr)$.

Before proceeding, we track the outgoing condition at $r=r_+$ to $z_\pm^{(\varsigma)}$. First, eq.~\eqref{EqPfypm} implies $y_+^{(+)}+y_-^{(+)}\in H\Delta^{-\frac34}\CI$, $y_+^{(+)}-y_-^{(+)}\in H\Delta^{-\frac14}\CI$ and $y_+^{(-)}+y_-^{(-)}\in H\Delta^{\frac54}\CI$, $y_+^{(-)}-y_-^{(-)}\in H\Delta^{\frac34}\CI$. For $X$ in~\eqref{EqPf2X}, we have
\begin{align*}
  &\adj(X)=\Delta^{\frac12}X_1+X_2,\ \ X^*=\Delta^{\frac12}X_1-X_2, \\
  &X_1=\begin{pmatrix} \frac{Q}{\bar\rho^*} & 0 \\ 0 & \frac{Q}{\bar\rho} \end{pmatrix},\ \ 
  X_2=\begin{pmatrix} -\frac{Q a\sin\theta}{\bar\rho^*} & i\bigl(M-\frac{Q^2}{\bar\rho}\bigr) \\ -i\bigl(M-\tfrac{Q^2}{\bar\rho^*}\bigr) & \frac{Q a \sin\theta}{\bar\rho} \end{pmatrix},
\end{align*}
and therefore eq.~\eqref{Eqzpm} yields
\begin{align*}
  z_+^{(\varsigma)}+z_-^{(\varsigma)} &= \Delta^{\frac12}X_1(y_+^{(\varsigma)}+y_-^{(\varsigma)}) + X_2(y_+^{(\varsigma)}-y_-^{(\varsigma)}), \\
  z_+^{(\varsigma)}-z_-^{(\varsigma)} &= X_2(y_+^{(\varsigma)}+y_-^{(\varsigma)}) + \Delta^{\frac12}X_1(y_+^{(\varsigma)}-y_-^{(\varsigma)}).
\end{align*}
This gives
\begin{equation}
\label{EqPf2Reg}
\begin{alignedat}{2}
  z_+^{(+)}{+}z_-^{(+)} &\in H\Delta^{-\frac14}\CI, &\ z_+^{(+)}{-}z_-^{(+)} &\in H\Delta^{-\frac34}\CI, \\
  z_+^{(-)}{+}z_-^{(-)} &\in H\Delta^{\frac34}\CI, &\ z_+^{(-)}{-}z_-^{(-)} &\in H\Delta^{\frac54}\CI.
\end{alignedat}
\end{equation}

\paragraph{Step~3: angular eigenvalue problem and separation.}

Since $\omega\in\R$, the operator $\cA^{(\varsigma)}$ is symmetric. In fact:

\begin{lemma}
\label{LemmaPf3SA}
  $\cA^{(\varsigma)}$ is essentially self-adjoint on $\CIc((0,\pi);\C^4)$, and $L^2((0,\pi),|\dd\theta|;\C^4)$ has a complete orthonormal eigenbasis of $\cA^{(\varsigma)}$.
\end{lemma}
\begin{proof}
  We drop the superscript $\varsigma$. The adjoint of $\cA\colon\CIc\subset L^2\to L^2$ has domain $\cD_{\rm max}=\{u\in L^2\colon\cA u\in L^2\}$. Suppose $u\in\cD_{\rm max}$ satisfies $(\cA^*-i)u=0$. Multiplying by $\sin\theta$, one obtains an ODE that is of regular-singular type at the poles $\theta=\theta_0\in\{0,\pi\}$; the indicial roots there are $\pm(\frac32\cos\theta_0+\varsigma m)$, each with a 2-dimensional space of indicial solutions. These roots are half-integers; and since $\theta^{-\frac12},(\pi-\theta)^{-\frac12}\notin L^2$, only the indicial root $\geq\frac12$ is allowed at either pole; so $|u|\lesssim\sqrt{\sin\theta}$ and $|\pa_\theta u|\lesssim\frac{1}{\sqrt{\sin\theta}}$, and we can thus deduce
\[
  0 = \Im\int_0^\pi \la(\cA^*-i)u,u\ra_{\C^4}\,\dd\theta = -\|u\|_{L^2}^2,
\]
so $u=0$. Similarly, one shows $\ker_{\cD_{\rm max}}(\cA^*+i)=0$. By \cite[Corollary to Theorem~VIII.3]{ReedSimonI}, this proves essential self-adjointness, with self-adjoint domain $\cD_{\rm max}$.

The existence of a complete orthonormal basis of eigenfunctions is a consequence of the compactness of the resolvent, which we prove as follows. Note that $x\cA$, $x:=\sin\theta$, is elliptic as a \emph{b-differential operator} \cite{MelroseAPS}, i.e., an operator constructed from $x\pa_\theta$ with coefficients in $\CI([0,\pi])$. Thus, $u\in L^2(|\dd\theta|)=x^{-\frac12}L^2(|\frac{\dd\theta}{\sin\theta}|)$ and $x\cA u\in x^{+\frac12}L^2(|\frac{\dd\theta}{\sin\theta}|)$, and the above information on indicial roots (called \emph{boundary spectrum} in b-terminology), imply $u\in x^\delta H^1_\bop([0,\pi];|\frac{\dd\theta}{\sin\theta}|)$ (the space of functions in $x^\delta L^2$ whose $x\pa_\theta$-derivative remains in $x^\delta L^2$) for all $\delta<\frac12$; so $\cD_{\rm max}\subset x^\delta H^1_\bop$. A variant of Rellich's theorem (e.g., \cite[Theorem~3.21]{HintzScaledBddGeo}) shows that for $\delta>-\frac12$, the inclusion $x^\delta H^1_\bop\hra x^{-\frac12}L^2(|\frac{\dd\theta}{\sin\theta}|)=L^2(|\dd\theta|)$ is compact.
\end{proof}

The regularity demanded of mode solutions implies $\vecz(r,\cdot),\pa_r\vecz(r,\cdot)\in L^2((0,\pi),|\dd\theta|)$, so~\eqref{EqPf2Final} gives $\cA\vecz(r,\cdot)\in L^2$, and hence $\vecz(r,\cdot)$ lies in the self-adjoint domain of $\cA$.

To decouple~\eqref{EqPf2Final}, we need the anticommutation relations (which follow from the form of $\Theta^{(\varsigma)}$)
\[
  \cA^{(\varsigma)}\cI = -\cI\cA^{(\varsigma)},\ 
  \cA^{(\varsigma)}\cJ = -\cJ\cA^{(\varsigma)},
\]
so for $\cS:=\cI\cJ$ we have $[\cA^{(\varsigma)},\cS]=0$. Since $\cS^2=-I_{4\times 4}$, the eigenvalues of $\cS$ are $\pm i$. For a joint ($\C^2\otimes\C^2$-valued) $(\cA^{(\varsigma)},\cS)$-eigenfunction $0\neq u=u(\theta)$ and $v:=\cI u$,
\begin{equation}
\label{EqPf3AngEig}
  \boxed{\cA^{(\varsigma)} u = \lambda u,\ \cS u = i u,}\quad \cA^{(\varsigma)} v=-\lambda v,\ \cS v=-i v;
\end{equation}
$u$ and $v$ are linearly independent (also when $\lambda=0$ since the $\cS$-eigenvalues of $u,v$ are distinct). The space of $\C^2\otimes\C^2$-valued functions of the form $\vecz^{(\varsigma)}(r,\theta) {=} z_1^{(\varsigma)}(r)(u(\theta){+}v(\theta)) {-} i z_2^{(\varsigma)}(r)(u(\theta){-}v(\theta))$ \footnote{The particular splitting here has the convenient effect of putting the terms $R_0^{(\varsigma)}$ in~\eqref{EqPf3SepFinal} below on the diagonal.}, $z_1^{(\varsigma)},z_2^{(\varsigma)}\in\CI((r_+,\infty);\C)$, is thus preserved by the operator in~\eqref{EqPf2Final}. We can write $L^2((0,\pi);\C^4)$ as the orthogonal direct sum of the spaces $\mathspan\{u,\cI u\}$ where $u$ ranges over an ONB of eigenfunctions of $\cA^{(\varsigma)}$ in $\ker(\cS-i)$ (thus $\cI u$ provides an ONB of $\ker(\cS+i)$ consisting of eigenfunctions of $\cA^{(\varsigma)}$). To prove the theorem, we thus only need to show $z_1^{(\varsigma)}=z_2^{(\varsigma)}=0$. The PDE~\eqref{EqPf2Final} reads
\begin{equation}
\label{EqPf3SepFinal}
  (\pa_r\mp\Delta^{-\frac12}R_0^{(\varsigma)})z_{1/2}^{(\varsigma)} = \Delta^{-\frac12}(R_1^{(\varsigma)}\mp i\lambda)z_{2/1}^{(\varsigma)};
\end{equation}
from~\eqref{EqPf2Reg} and $z_1^{(\varsigma)}(u+v)=\frac12(I_{4\times 4}+\cI)\vecz^{(\varsigma)}$ and $z_2^{(\varsigma)}(u-v)=\frac{i}{2}(I_{4\times 4}-\cI)\vecz^{(\varsigma)}$, the horizon regularity is
\begin{equation}
\label{EqPf3Reg}
\begin{alignedat}{2}
  z_1^{(+)} &\in H\Delta^{-\frac14}\CI,&\ z_2^{(+)} &\in H\Delta^{-\frac34}\CI, \\
  z_1^{(-)} &\in H\Delta^{\frac34}\CI,&\ z_2^{(-)} &\in H\Delta^{\frac54}\CI.
\end{alignedat}
\end{equation}

Define
\begin{equation}
\label{EqPf3xi0}
  \wt K := K - \xi_0\Delta,\quad \xi_0 := \tfrac{2\omega Q^2}{M^2-Q^2},
\end{equation}
and note that $\Delta^{-\frac12}R_0^{(\varsigma)}=\frac{3\Delta'}{4\Delta}-i\varsigma\frac{\wt K}{\Delta}$ is the logarithmic derivative of $\Delta^{\frac34}H^{-\varsigma}e^{i\varsigma\xi_0 r}$; \eqref{EqPf3Reg} thus suggests defining
\begin{align}
  &w^{(+)}{:=}\Delta^{\frac34}H^{-1}e^{i\xi_0 r}z_2^{(+)},\,
  w^{(-)}{:=}\Delta^{-\frac34}H^{-1}e^{i\xi_0 r}z_1^{(-)} {\in} \CI; \nonumber\\
\label{EqPf3ODEVars}
  &\quad \pa_r w^{(\varsigma)}=(R_1^{(\varsigma)}+i\varsigma\lambda)w^{(\varsigma)}_\sharp,
\end{align}
$w^{(+)}_\sharp:=\Delta^{\frac14}H^{-1}e^{i\xi_0 r}z_1^{(+)}$, $w^{(-)}_\sharp:=\Delta^{-\frac54}H^{-1}e^{i\xi_0 r}z_2^{(-)}$. Taking another $r$-derivative and using~\eqref{EqPf3SepFinal} gives
\begin{equation}
\label{EqPf3wsharp}
  \Delta\pa_r w^{(\varsigma)}_\sharp+\bigl(\tfrac{1-3\varsigma}{2}\Delta'+2 i\wt K\bigr)w^{(\varsigma)}_\sharp=(R_1^{(\varsigma)}-i\varsigma\lambda)w^{(\varsigma)}.
\end{equation}
We arrive at two decoupled second-order radial ODEs:
\begin{widetext}
\begin{equation}
\label{EqPf3ODE}
  \boxed{\cL^{(\varsigma)}w^{(\varsigma)} = 0,\quad
  \cL^{(\varsigma)}:=\Bigl(\pa_r - \Bigl[\frac{(1+3\varsigma)\Delta'}{2\Delta} - \frac{2 i\wt K}{\Delta} + \frac{\pa_r R_1^{(\varsigma)}}{R_1^{(\varsigma)}+i\varsigma\lambda}\Bigr]\Bigr) \Delta\pa_r - \bigl((R_1^{(\varsigma)})^2+\lambda^2\bigr),\quad \varsigma=\pm 1.}
\end{equation}
\end{widetext}
This is regular on $(r_+,\infty)$ and has $\CI([r_+,\infty))$-coefficients. (Note: $\Re(R_1^{(\varsigma)}+i\varsigma\lambda)=-\varsigma\frac{2\omega Q(M r-Q^2)}{M^2-Q^2}\neq 0$.)

\paragraph{Step~4: passage to an adjoint equation; Fourier transform.}

\emph{We drop the superscript $\varsigma=\pm 1$.} It remains to prove that outgoing solutions $w$ of $\cL w=0$ vanish. We first transform $w$ to a solution of the \emph{bilinear adjoint} $\cL^\dag$ (i.e., $\int \cL w\,v\,\dd r=\int w\,\cL^\dag v\,\dd r$). The square bracket in~\eqref{EqPf3ODE} being the logarithmic derivative of $\Delta^{\frac{1+3\varsigma}{2}}H^{-2}e^{2 i\xi_0 r}(R_1+i\varsigma\lambda)=:f^{-1}$,
\begin{align*}
  &\tilde w := f w\in(r-r_+)^\alpha\CI, \ \alpha:=-\tfrac{1+3\varsigma}{2}+\tfrac{2 i K(r_+)}{\Delta'(r_+)} \\
  &\implies \cL^\dag \tilde w = f\cL w=0.
\end{align*}
Since $K$ is real, $\alpha\in\C\setminus(-\N_0)$ unless $\varsigma=+1,K(r_+)=0$ (so $\alpha=-2$).

\paragraph{4.1. The generic case $\alpha\neq -2$.} Extending the homogeneous terms $(r-r_+)^{\alpha+j}$ for $j\in\N_0$, $\Re\alpha+j\leq 1$, as $(r-r_+)_+^{\alpha+j}$ produces the canonical distributional extension of $\tilde w$ by $0$ to $\R\ni r<r_+$. The algebra simplifies for $v=(R_1+i\varsigma\lambda)\tilde w\in\sS'(\R)$, $\supp v\subset[r_+,\infty)$, which solves $\tilde\cL^\dag v=0$, $\tilde\cL^\dag := \pa_r(R_1+i\varsigma\lambda)^{-1}\bigl(\Delta\pa_r+\bigl[\frac{1+3\varsigma}{2}\Delta'-2 i\wt K\bigr]\bigr)-(R_1-i\varsigma\lambda)$ \footnote{The source term is supported at $r=r_+$ and, unless $\varsigma=-1,K(r_+)=0$ (in which case $\alpha=1$), has non-integer homogeneity with respect to scaling of $r-r_+$, and hence vanishes. For $\alpha=1$, the source term is supported at $r=r_+$ but has non-negative homogeneity, and thus again vanishes.}. The indicial roots of $\Delta\tilde\cL^\dag$ are $0,\alpha$, with $0-\alpha\notin\N_0$; therefore, once we show that the leading-order coefficient
\begin{equation}
\label{EqPf4Coeff}
  V_\cH := \lim_{r\searrow r_+} \bigl((r-r_+)^{-\alpha}v(r)\bigr) \in \C
\end{equation}
vanishes, we conclude that $v$ vanishes to infinite order at $r=r_+$, and thus $v=0$ on $(r_+,\infty)$ by unique continuation; this then gives $w=0$ (thus also $w_\sharp=0$).

Clearing denominators yields an ODE with coefficients that are \emph{cubic} polynomials in $r$; differentiating in $r$, one finds that
\begin{align*}
  (R_1{+}i\varsigma\lambda)^{-1}\pa_r (R_1{+}i\varsigma\lambda)^2\tilde\cL^\dag &{=} \pa_r^2\bigl(\Delta\pa_r{+}\bigl[\tfrac{1{+}3\varsigma}{2}\Delta'{-}2 i\wt K\bigr]\bigr) \\
  &\hspace{-2em}{-} (R_1^2{+}\lambda^2)\pa_r {-} (3 R_1{-}i\varsigma\lambda)(\pa_r R_1)
\end{align*}
has \emph{quadratic} coefficients (and annihilates $v$). We now pass to the Fourier transform in $r$, so $\pa_r,r\rightsquigarrow i\xi,i\pa_\xi$; thus $\hat v=\hat v(\xi)$ satisfies the second-order ODE
\begin{equation}
\label{EqPf4FT}
  a(\xi)\hat v'' + (b(\xi)+i c(\xi))\hat v' + (d(\xi)+i e(\xi))\hat v = 0,
\end{equation}
where $a,b,c,d,e$ are the \emph{real-valued} functions
\begin{align*}
  a(\xi) &= \xi(\xi+\xi_0)(\xi+\xi_0+2\omega), \\
  b(\xi) &= \xi^2-\xi_0(\xi_0+2\omega) - 3\varsigma\xi(\xi+\xi_0+2\omega), \\
  c(\xi) &= 2 M\xi(\xi+\xi_0)^2, \\
  d(\xi) &= \bigl(2 a(m-a\omega)-(a^2+Q^2)(\xi+2\xi_0)\bigr)\xi^2 \\
    &\quad + \bigl(\tfrac{9 M^2}{4 Q^2}-Q^2\xi_0^2-\lambda^2\bigr)\xi + \bigl(\tfrac{3\varsigma}{2}(\xi_0+2\omega) - \tfrac{M}{Q}\xi_0\lambda\bigr), \\
  e(\xi) &= M(\xi+\xi_0)(\xi-\xi_0-3\varsigma\xi).
\end{align*}
The function $f=\xi^{-2}(\xi+\xi_0)^{-3\varsigma}$ has logarithmic derivative $\frac{b-a'}{a}$, and thus we can eliminate the real $\hat v'$-coefficient in~\eqref{EqPf4FT} via $0=(f a\hat v')'+i(f c\hat v'+f e\hat v)+f d\,\hat v$. In view of the algebraic relation $(\frac12 f c)'=f e$, we conclude that~\eqref{EqPf4FT} can be written as
\begin{equation}
\label{EqPf4FT2}
\begin{split}
  &(p\hat v')' + i(q\pa_\xi+\pa_\xi q)\hat v + \tilde d\,\hat v=0, \\
  &\quad p:=\tfrac{\xi+\xi_0+2\omega}{\xi(\xi+\xi_0)^{3\varsigma-1}},\ \ q:=M\xi^{-1}(\xi+\xi_0)^{2-3\varsigma},
\end{split}
\end{equation}
with $\tilde d=\tilde d(\xi)$ real-valued.

Since $v$ vanishes for $r<r_+$, has the leading-order singularity $(r-r_+)_+^\alpha V_\cH$ at $r=r_+$, is smooth on $(r_+,\infty)$, and of class $H^2 e^{-2 i\xi_0 r}\times H^{-1}e^{i\xi_0 r}\times H^{-1}=e^{-i\xi_0 r}$ times $\cO(r^C)$ as $r\to\infty$, $\hat v(\xi)$ is $\CI$ except at $\xi=-\xi_0$, and
\[
  \hat v(\xi) = e^{-i r_+\xi}|\xi|^{-\alpha-1}\bigl( c_\alpha^\pm V_\cH + \cO(|\xi|^{-1})\bigr),\quad \xi\to\pm\infty,
\]
where $c_\alpha^\pm:=\Gamma(\alpha+1)e^{\mp\frac{i\pi}{2}(\alpha+1)}\neq 0$. Since $p,q,\tilde d$ are real,~\eqref{EqPf4FT2} gives the constancy of the modified Wronskian
\[
  \sW(\xi) := p \Im\bigl(\hat v'\ol{\hat v}\bigr) + q|\hat v|^2
\]
on each connected component of $\R\setminus\{0,-\xi_0\}$; we compute
\begin{align*}
  \sW({-}\xi_0{-}2\omega) &{=} {-}M(-2\omega)^{2-3\varsigma}(\xi_0{+}2\omega)^{-1}|\hat v(-\xi_0{-}2\omega)|^2, \\
  \lim_{\xi\to\pm\infty}\sW(\xi) &{=} {-}(r_+{-}M)|c_\alpha^\pm|^2|V_\cH|^2.
\end{align*}
Recall that $\xi_0$ is a positive multiple of $\omega$. When $\omega<0$, we get $0\leq\sW(-\xi_0-2\omega)=\sW(+\infty)\leq 0$, so $\sW(+\infty)=0$. When $\omega>0$, then $0\leq\sW(-\xi_0-2\omega)=\sW(-\infty)\leq 0$, so $\sW(-\infty)=0$. Either way, $V_\cH=0$. As argued after~\eqref{EqPf4Coeff}, this completes the proof for $\omega\neq 0$ unless $\alpha=-2$.

\paragraph{4.2. The exceptional case $\alpha=-2$.}

Writing $v(r)=(r-r_+)^{-2}V_\cH+(r-r_+)^{-1}V^\flat_\cH+\CI([r_+,\infty))$, $\tilde\cL^\dag v=0$ implies a homogeneous linear system for $(V_\cH,V^\flat_\cH)\in\C^2$ that coincides with the system guaranteeing $\tilde\cL^\dag(V_\cH\chi_+^{-2}(r-r_+)-V^\flat_\cH\chi_+^{-1}(r-r_+))=0$ where $\chi_+^z(x)=\frac{x_+^z}{\Gamma(z+1)}$, analytically extended from $\Re z>-1$ to $z\in\C$. (This is an instance of the strong intertwining property \cite[Proposition~5.4]{PetersenVasyModeStab}.) Since $(x\pa_x+2)\delta'(x)=0$, the operator $\tilde\cL^\dag$ preserves the space $\cV:=\mathspan\{\delta(r-r_+),\delta'(r-r_+)\}$. Applying the previous arguments to $v\in\cV$ (and using that $-1$ is not an indicial root) shows that $\tilde\cL^\dag|_\cV\colon\cV\to\cV$ is injective (hence surjective). Therefore, $V_\cH=V^\flat_\cH=0$, and hence $v$ is, in fact, smooth on $[r_+,\infty)$. Its extension by $0$ to $r<r_+$ thus satisfies $\tilde\cL^\dag v\in\cV$, so we can find $v_\flat=c_1\delta'+c_0\delta\in\cV$ with $v+v_\flat\in\ker\tilde\cL^\dag$. The Wronskian argument gives $c_1=0$; the distributional equation $\tilde\cL^\dag(v+c_0\delta)=0$ then implies $c_0=0$, so $\tilde\cL^\dag v=0$; and this, with $v(r)=v(r_+)\bOne_{[r_+,\infty)}(r)+\cO((r-r_+)_+)$, implies $v(r_+)=0$. Since all indicial roots of $\Delta\tilde\cL^\dag$ are $\leq 0$, this now yields the infinite order vanishing of $v$ at $r=r_+$, and thus again $v=0$ globally.

\paragraph{Step~5. Analysis for $\omega=0$.}
\label{SsPf0}

The transformations and separation in Steps~1--3 apply; from~\eqref{EqPf3ODEVars}--\eqref{EqPf3wsharp} we get
\begin{equation}
\label{EqPf0w}
  \bigl[\pa_r\Delta\pa_r {-} \bigl(\tfrac{1+3\varsigma}{2}\Delta'{-}2 i a m\bigr)\pa_r {+} \bigl(\tfrac{9 M^2}{4 Q^2}{-}\lambda^2\bigr)\bigr] w^{(\varsigma)} = 0.
\end{equation}
The regularity at the event horizon is as in~\eqref{EqPf3ODEVars}; at infinity, we have $y_\pm^{(\varsigma)}=\cO(r^{-\frac12})$ and thus $w^{(+)}=\cO(r)$, $w^{(+)}_\sharp=\cO(1)$, $w^{(-)}=\cO(r^{-2})$, $w^{(-)}_\sharp=\cO(r^{-3})$.

\paragraph{5.1. The case $\varsigma=-1$.}
\label{SssPf0m1}

For $w:=w^{(-)}$, eq.~\eqref{EqPf0w} gives
\[
  \sW_0' = a m\Delta'|w|^2,\ 
  \sW_0:=\Delta^2\Im\bigl( w'\ol{w}\bigr) + a m\Delta|w|^2.
\]
Since $\sW_0(r)=0$ for $r=r_+,\infty$, we get $\int_{r_+}^\infty a m\Delta'|w|^2\,\dd r=0$, so $w=0$ when $a m\neq 0$. For $a m=0$, we multiply~\eqref{EqPf0w} by $\bar w$ and integrate, getting $\int_{r_+}^\infty \Delta|w'|^2{+}(\lambda^2{-}\frac{9 M^2}{4 Q^2}{+}1)|w|^2\,\dd r{+}\frac{\Delta'(r_+)}{2}|w(r_+)|^2{=}0$. To conclude $w=0$, we use:

\begin{lemma}
\label{LemmaPf0Ang}
  All eigenvalues $\lambda$ of $\cA^{(\varsigma)}|_{\omega=0}$ satisfy $|\lambda|\geq\frac{3 M}{2|Q|}$. For eigenfunctions in $\ker(\cS\mp i)$, $\lambda\neq\pm\frac{3 M}{2 Q}$.
\end{lemma}
\begin{proof}
  The action of $\cI$ flips the sign of eigenvalues of $\cS$ and $\cA^{(\varsigma)}$. It thus suffices to consider the action of $\cA^{(\varsigma)}$ on $\ker(\cS-i)$. Using the basis $(i,1,0,0)^T,(0,0,-i,1)^T$ of $\ker(\cS-i)\subset\C^4$, this action is given by $A+B$ where
  \begin{align*}
    &A:=\begin{pmatrix} i\pa_\theta & -i\Theta_0^{(\varsigma)} \\ i\Theta_0^{(\varsigma)} & -i\pa_\theta \end{pmatrix},\ \ B:=\begin{pmatrix} 0 & \Theta_1 \\ \Theta_1 & 0 \end{pmatrix}, \\
    &\qquad \Theta_0^{(\varsigma)}=\frac32\cot\theta+\frac{\varsigma m}{\sin\theta},\ \ \Theta_1:=-\frac{3 M}{2 Q}.
  \end{align*}
  Note that $A B=-B A$ and $B^2=\frac{9 M^2}{4 Q^2}I_{2\times 2}$. For $u\in\CIc((0,\pi);\C^2)$ with $\|u\|_{L^2}=1$, it follows that
  \[
    \|(A+B)u\|_{L^2}^2 = \|A u\|_{L^2}^2 + \tfrac{9 M^2}{4 Q^2}\|u\|_{L^2}^2\geq \tfrac{9 M^2}{4 Q^2}.
  \]
  This remains valid on the self-adjoint domain of $A+B$, proving the eigenvalue bound. Equality, for an eigenfunction $u=u_+(1,1)^T+u_-(1,-1)^T$, forces $A u=0$ and thus $\pa_\theta u_\pm\mp\Theta_0^{(\varsigma)}u_\pm=0$; so $u_-$ is a scalar multiple of $(\sin\theta)^{-\frac32}(\tan\frac{\theta}{2})^{-\varsigma m}\notin L^2$ and must vanish. Thus, $u=u_+(1,1)^T$. The eigenvalue equation $B u=\lambda u$ then gives $\lambda=\Theta_1$.
\end{proof}

From Step~3 onward, we used angular eigenfunctions in $\ker(\cS{-}i)$. Since $\lambda\neq\frac{3 M}{2 Q}$, eq.~\eqref{EqPf3ODEVars}, i.e., $0=i(\frac{3 M}{2 Q}{-}\lambda)w_\sharp^{(-)}$, gives $w_\sharp^{(-)}=0$.

\paragraph{5.2. The case \texorpdfstring{$\varsigma=+1$}{varsigma=+1}.}

For $w:=w^{(+)}$, we can effect a sign flip of the $\Delta'$-term in~\eqref{EqPf0w} by thrice differentiating in $r$. (This idea is related to the enhanced red-shift effect \cite[\S{7}]{DafermosRodnianskiLectureNotes}.) This gives
\[
  \bigl[ \pa_r\Delta\pa_r + (\Delta'+2 i a m)\pa_r + \bigl(\tfrac{9 M^2}{4 Q^2}-\lambda^2\bigr)\bigr] w''' = 0.
\]
Since $w'''=\cO(r^{-2})$ (like $w^{(-)}$ previously), the arguments in Step~5.1 give $w'''=0$. The requirement $w=\cO(r)$ gives $w=A r+B$, $A,B\in\C$. Substituting into~\eqref{EqPf0w}, the coefficient of $r$ equals $A$ times $-2+\frac{9 M^2}{4 Q^2}-\lambda^2<0$, forcing $A=0$. Insert this into eq.~\eqref{EqPf3wsharp} for $w_\sharp:=w_\sharp^{(+)}$; differentiating this equation thrice in $r$ gives $(\pa_r\Delta\pa_r+(\Delta'+2 i a m)\pa_r)w''_\sharp=0$, hence $w''_\sharp=0$, and the a priori bound $w_\sharp=\cO(1)$ yields $w_\sharp=C\in\C$. Eq.~\eqref{EqPf3wsharp} now reads $(-2 r+2 M+2 i a m)C=i(\frac{3 M}{2 Q}-\lambda)B$. The coefficient of $r$ gives $C=0$. By Lemma~\ref{LemmaPf0Ang}, $\lambda\neq\frac{3 M}{2 Q}$, so $B=0$, and hence $w=0$.\hfill$\blacksquare$

\prlheading{Conclusion and outlook}%
We separated the perturbation equations on KN into angular and radial ODEs, and proved that the radial ODEs do not admit purely oscillatory outgoing mode solutions \footnote{Direct proofs of mode stability for Kerr or Reissner--Nordstr{\"o}m in $\Im\omega>0$ in the literature rely on a full separation into angular and radial ODEs. For rotating BHs and $\omega\notin\R$, the angular ODE is not self-adjoint, so it is not clear that one can decompose a putative unstable mode solution into a sum of separated mode solutions (see however \cite{FinsterSmollerSpheroidalCompleteness}). Mode stability on the real axis and a continuity argument are a robust route to proving full mode stability.}. To deduce from this the mode stability of subextremal KN \emph{for the linearized Einstein--Maxwell system}, one needs to extend Wald's result \cite{WaldKerrPerturbation0}: the vanishing of $\psi_{\pm 2},\psi_{\pm 1}$ can be shown to imply that a mode solution at frequency $\omega\in\R$ is a pure gauge term plus the linearization of KN in the BH parameters. Details will be given elsewhere. (See also \cite{LeeChargedBHPert,LeeMcGlinnRNUniq,DebeverKamranMcLenaghanTypeDSol,CrossmanFackerellRotatingBHPert,FackerellKNLin,LunKNPert,WellerEtAlBHIsospec,HollandsToomaniKNMaxwell}.) The author expects that \cite{HintzKerrStab} can then be adapted to upgrade mode stability to the full nonlinear stability of subextremal KN.

\vspace{0.5em}
\prlheading{Acknowledgments}%
The author gratefully acknowledges support from the NSF under grant DMS-2554160.

\vspace{0.5em}
\prlheading{Data availability statement}%
No data were created or analyzed in this study.

\vspace{0.5em}
\prlheading{AI use}%
The author proposed several proof strategies to GPT-6 Pro, which assisted with testing these approaches through exploratory calculations, algebraic checks, literature review, and the development of draft arguments. The author substantially reworked these and wrote the final manuscript.

%


\end{document}